\documentclass[letterpaper, 10 pt, conference]{ieeeconf}  % Comment this line out if you need a4paper

\IEEEoverridecommandlockouts                              % This command is only needed if 
\usepackage{graphicx} % for pdf, bitmapped graphics files
\usepackage{amsmath} % assumes amsmath package installed
\usepackage{amssymb}  % assumes amsmath package installed
\usepackage{siunitx}
\usepackage{algorithm}
\usepackage{algpseudocode}
\usepackage{setspace}

\usepackage{color}
\usepackage{tikz}

\newcommand\copyrighttext{%
  \centering
  Accepted for the 20th European Control Conference 2022
  }
\newcommand\copyrightnotice{%
\begin{tikzpicture}[remember picture,overlay]
\node[anchor=south,yshift=10pt] at (current page.south) {\fbox{\parbox{\dimexpr\textwidth-\fboxsep-\fboxrule\relax}{\copyrighttext}}};
\end{tikzpicture}%
}

\newtheorem{lemma}{Lemma}
\newtheorem{prop}{Proposition}

\DeclareMathOperator*{\argmin}{arg\,min}
\newcommand{\diff}{\mathrm{d}}
\newcommand{\e}{\mathrm{e}}

\graphicspath{{figsArxiv/}}

\title{\LARGE \bf
Input Sequence and Parameter Estimation in Impulsive Biomedical Models*
}

\author{H{\aa}kan Runvik$^{1}$ and Alexander Medvedev$^{1}$% <-this % stops a space
\thanks{*This work is funded in part by the PhD program at the Centre for Interdisciplinary Mathematics, Uppsala University, Sweden and the Swedish Research Council Grant 2019-04451 for the project ``Synchronization and
entrainment in the impulsive Goodwin’s oscillator".}% <-this % stops a space
\thanks{$^{1}$H. Runvik and A. Medvedev are  with the department of Information Technology, Uppsala University, Uppsala, Sweden,
        {\tt\small \{hakan.runvik, alexander.medvedev\}@it.uu.se}}%
}

\begin{document}

\maketitle
\copyrightnotice

\thispagestyle{empty}
\pagestyle{empty}

%%%%%%%%%%%%%%%%%%%%%%%%%%%%%%%%%%%%%%%%%%%%%%%%%%%%%%%%%%%%%%%%%%%%%%%%%%%%%%%%
\begin{abstract}
A hybrid model for biomedical time series comprising a continuous second-order linear time-invariant system driven by an  input sequence of positively weighted Dirac delta-functions is considered. 
The problem of the joint estimation of the input sequence and the continuous system parameters from output measurements is investigated. A solution that builds upon and refines a previously published least-squares formulation is proposed. Based on a  thorough analysis of the properties of the least-squares solution, improvements in terms of accuracy and ease of use are achieved on synthetic data, compared to  the original algorithm.
\end{abstract}

%%%%%%%%%%%%%%%%%%%%%%%%%%%%%%%%%%%%%%%%%%%%%%%%%%%%%%%%%%%%%%%%%%%%%%%%%%%%%%%%
\section{INTRODUCTION}
Signals exhibiting slow, dissipative dynamics that are interrupted by multiple rapid bursts occur in many biological systems. Common examples are found in e.g.  endocrinology, since pulsatility is recognized as a fundamental property in the secretion of most hormones \cite{vkp08}. In pharmacokinetics, multi-peaking phenomena in drug concentration \cite{dt10} can also display such characteristics. 
There is no generally accepted approach to mathematical modeling of these behaviors.
In the endocrine case, a popular construct features a linear plant to portray the hormone elimination fed with an input signal that represents the secretion episodes. For instance, a Gaussian  input signal shape is assumed in \cite{jp09}, which enables deconvolution-based input estimation.

To avoid additional assumptions, a pulsatile time series is modeled by a linear plant with impulsive input in this work. In closed loop, this setup was developed for modeling testosterone regulation in \cite{cms09}, while a similar model was employed for pharmacokinetic applications in \cite{rmk20}. The estimation of the input sequence and continuous plant parameters is treated. Least-squares (LS) methods were previously used to address this hybrid identification problem \cite{hm10}, \cite{MM:13}, while a Laguerre domain approach was employed for the input estimation in \cite{rm20}. The present work is based on the same optimization formulation as \cite{MM:13}, where LASSO (least absolute shrinkage and selection operator) regularization was used. Yet, a more rigorous estimation procedure is achieved based on a comprehensive analysis of the underlying optimization problem. The main contribution is in the optimization problem analysis that underpins the theoretical foundation of the identification approach. Further, the resulting estimation method does not require user-defined data-dependent parameters and displays better performance as well as ease of implementation.

The rest of the paper is organized as follows. First, the model and estimation problem are formulated. Then, an analysis of the parameter-dependent characteristics of the LS solution is performed and shown to enable an efficient estimation of the parameters in the noise-free case. Finally, the method is generalized to account for noise and uncertainties and experimental results for synthetic data are presented.

\section{ESTIMATION PROBLEM}
\subsection{Model description}
Consider the impulsive sequence 
\begin{equation} \label{eq:input}
    \xi(t)=\sum_{k=0}^{\infty}d_k \delta\left(t-\tau_k\right),
\end{equation}
where $\delta(\cdot)$ is the Dirac delta function and $d_k$ and $\tau_k$ determine the positive impulse weights and times. It is fed into a linear time-invariant compartmental state-space model 
\begin{equation} \label{eq:sys}
    \dot{x}=Ax+B\xi(t), \quad y = Cx, \quad x=\begin{bmatrix}x_1 &x_2 \end{bmatrix}^\intercal,
\end{equation}
where
\begin{equation} \label{eq:sysmatr}
    A= \begin{bmatrix}
    -b_1 & 0 \\
    g_1 & -b_2
    \end{bmatrix}, \quad B=\begin{bmatrix}
    1 \\ 0
    \end{bmatrix}, \quad C=\begin{bmatrix} 0 \\ 1 \end{bmatrix} ^\intercal,
\end{equation}
with positive parameters $b_1$, $b_2$, $g_1$. 
Defining the Heaviside step function as $H(t)$ and assuming the initial state $x(t_0)=x_0$, the output of the system can in a straightforward manner be calculated as
\begin{multline*}
    y(t) =C\Big( e^{A(t-t_0)}x_0 + \int_{t_0}^t \e^{A(t-\tau)}B \xi(\tau)~\diff \tau \Big)\\
    = C e^{A(t-t_0)}x_0 + \sum_{k=0}^\infty g_1 d_k z(b_1, b_2, t-\tau_k),
\end{multline*}
where
\begin{equation*}
    z(b_1, b_2, t) = \frac{\e^{-b_2 t}-\e^{-b_1 t}}{b_1-b_2}H(t).
\end{equation*}

\subsection{Estimation problem formulation}
Let the output of \eqref{eq:sys} be sampled, possibly irregularly, over a finite horizon and result in the measurements $y(t_k)$, where $k=1,\dots,K$ and $t_k<t_{k+1}$, thus yielding the vector
\begin{equation*}
    Y = \begin{bmatrix} y(t_1) & \dots & y(t_K) \end{bmatrix}^\intercal.
\end{equation*}
Since an impulse in between two sampling times cannot be distinguished in the sampled output from a pair of impulses that occur at the sampling times \cite{MM:13}, the impulses are without loss of generality restricted to occur at the sampling times. 
Then it holds that
\begin{equation}\label{eq:matrixY}
    Y = \Phi(b_1, b_2) \theta,
\end{equation} 

where
\begin{equation*}
    \Phi(b_1, b_2)=\begin{bmatrix} \varphi(b_1, b_2, t_1) & \dots & \varphi(b_1, b_2, t_K) \end{bmatrix}^\intercal,
\end{equation*}

\begin{equation*}
    \varphi(b_1, b_2, t_i)=\begin{bmatrix}\e^{-b_2(t_i-t_1)}\\ z(b_1, b_2, t_i-t_1)\\ \vdots \\ z(b_1, b_2, t_i-t_K) \end{bmatrix},
\end{equation*}

\begin{equation*}
    \theta=\begin{bmatrix}x_2(t_1) & d_1 & \dots & d_{K-1} \end{bmatrix}^\intercal.
\end{equation*}
Notice that $\Phi(b_1, b_2)$ is square and that the state $x_2(t_1)$ is included in the formulation rather than $x_0$, since $x_2(t_1)$ and $d_1$ uniquely determine the state of the system for $t>t_1$.

Further, the combined impulse and parameter estimation in system \eqref{eq:sys} is treated, i.e. the parameters $d_k, k=1,\dots,K$ and $b_i, i=1,2$ are sought. Notice that $g_1=1$ can be assumed, as changing this parameter corresponds to scaling of the impulses. Furthermore, assume that $b_1<b_2$. 

An LS optimization formulation introduced in \cite{MM:13} is employed to the problem in hand. In the estimation,  $b_i^*, i=1,2$ denote the true parameter values 
while $b_i$ represent the parameters in the LS formulation 
\begin{equation}
\begin{aligned} \label{eq:opt}
\hat\theta(b_1, b_2) = \argmin_\theta & \|Y - \Phi(b_1, b_2) \theta \|^2 ,
%\\
    %\mathrm{s.t.} & \quad \theta \ge 0,
\end{aligned}
\end{equation}
where $\|\cdot\|$ is the Euclidean vector norm. 
The parameter-dependent objective function makes the setup resemble a multi-parametric programming problem (see e.g. \cite{od16}).
We also use the notation $d_k^*$ for the true impulse weights while $d_k$ represent their estimates for given parameter values $b_1$ and $b_2$ (we suppress the dependency for ease of notation).

In the noise-free case, the optimization formulation is unconstrained. The estimate of $\theta$ can therefore be calculated via a  matrix inversion (the invertibility of $\Phi$ is shown in \cite{MM:13}). In the presence of noise or uncertainties, the impulse weights are restricted to be non-negative, which results in a constrained LS problem.

\subsection{Estimation principle}\label{sec:technique}
In the noise-free case, the estimation is based on the following properties of the optimization problem in \eqref{eq:opt}, which are given in Proposition~\ref{prop:impsign}.
\begin{itemize}
    \item If $b_1+b_2>b_1^*+b_2^*$ and $b_1>b_1^*$, all impulse estimates solving \eqref{eq:opt} have positive weights;
    \item If $b_1+b_2<b_1^*+b_2^*$ and $b_1<b_1^*$, all impulse estimates solving \eqref{eq:opt} that do not correspond to true impulses have negative weight.
\end{itemize}
The properties above give rise to the division of the parameter space depicted in the left subplot of Fig.~\ref{fig:optSol}. The idea is to utilize the structure of this space to identify $b_1^*$ and $b_2^*$. A problem arises in the regions marked as unknown, where the signs of the impulse weights solving the optimization problem vary depending on the data. However, theoretical reasoning regarding simplified cases (see Section~\ref{sec:boundary}) indicates that the parameter space will be divided qualitatively according to the right subplot of Fig.~\ref{fig:optSol}. This is also observed in numerical experiments. The quantitative behavior (i.e. the slopes of the curves) depends on the values of $b_1^*,b_2^*$ and the distribution of impulses and sampling instances. This refined partitioning enables navigation in the parameter space to the point $(b_1^*,b_2^*)$, as described in Section~\ref{sec:nonoiseest}.

\begin{figure}
\footnotesize
\begin{center}
\def\svgwidth{0.515\textwidth}
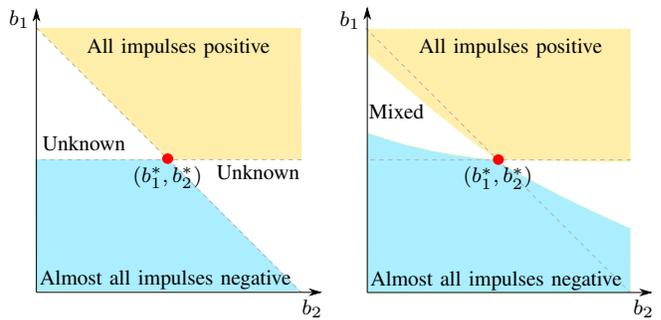
\caption{The division of the $b_1$-$b_2$ parameter space, relative to the true parameter values $(b_1^*,b_2^*)$, in terms of the solutions to \eqref{eq:opt}. Left: regions with guaranteed signs of impulse weights. Right: region boundaries according to the analysis in Section \ref{sec:boundary} and numerical experiments. ``Mixed'' indicates that both positive and negative impulses are present.}
\label{fig:optSol}
\end{center}
\end{figure}

When noise or uncertainties are present, the parameter space is no longer divided in well-defined regions, as impulses of both signs typically appear in most of the parameter space. However, it is still possible to estimate the boundary of the region with positive impulses, which we denote $\gamma_P$. If the noise level is low enough, it is also possible to find an approximation of the true parameter values.

\section{NOISE-FREE ESTIMATION}\label{sec:noisefree}
We will first show how the $b_1$-$b_2$ parameter space is divided into the regions indicated above. The analysis is based on the following lemma.

\begin{lemma}\label{lem:intersect}
Let $y(t)$ be the response of system \eqref{eq:sys} to the input $d_1\delta(t-t_1)$ with $x_0 = \begin{bmatrix}0&0\end{bmatrix}^\intercal$. Denote the response of an estimate of \eqref{eq:sys} with the parameters $\hat b_1$, $\hat b_2$ to the input $\hat d_1\delta(t-\hat t_1)$  as $\hat y(t)$. Assume that $\hat b_1+\hat b_2>b_1+b_2$, $\hat{b}_1>b_1$ and $\hat{b}_1>\hat b_2$. Then, there exist at most two $\tau>\max \{t_1, \hat t_1\}$ such that $y(\tau)=\hat y(\tau)$.
\end{lemma}
\begin{proof}
See Appendix \ref{app:intersect}.
\end{proof}

Provided that the impulse sequence is sparse (i.e. that the set $S$ below is nonempty), the result above can be used to characterize the properties of the solution to \eqref{eq:sys} in two cases, where the estimated system is either faster or slower than the true dynamics. These are defined in the following proposition.

\begin{prop} \label{prop:impsign}
Consider LS problem \eqref{eq:opt} with the initial condition $x_2(t_1)=0$.
Assume that the noise-free measurements $Y$ are produced by \eqref{eq:sys}   and let $S=\{k\in \{1,\dots,K\} \mid d_k^*=0\}$. If $b_1+b_2>b_1^*+b_2^*$ and $b_1>b_1^*$, then $d_k>0$ for $k=1,\dots K$. If $b_1+b_2<b_1^*+b_2^*$ and $b_1<b_1^*$, then $d_k<0$ for $k\in S$.
\end{prop}
\begin{proof}
See Appendix~\ref{app:impsign}.
\end{proof}

Note that Lemma~\ref{lem:intersect} and Proposition~\ref{prop:impsign}  apply only when the initial state is zero, while the optimization formulation in \eqref{eq:opt} allows a nonzero initial value for $x_2$ (a nonzero initial $x_1$ can be represented by an impulse and is thus not included in the estimation). However, since the contribution from the initial state tends to zero exponentially, the proposition is expected to hold in the case of nonzero initial conditions too.

\subsection{Boundaries of the sign-definite impulse regions}\label{sec:boundary} 
The proposition above does not provide information about the region which is marked as unknown in the left subplot of Fig.~\ref{fig:optSol}. To gain understanding of the behavior in this region, consider a simplified case of three sampled measurements of the response to a single impulse at time $t=0$. Denote it as $y^*(t)$ and assume that it is generated by \eqref{eq:sys} with the parameter values $b_1^*$, $b_2^*$. Let $y(t)$ be the response of the same system but with parameters $b_1$, $b_2$. The boundary between solutions with positive and negative impulse weights is then defined by the case when $y(t)$ intersects $y^*(t)$ precisely at the sampling times (i.e. no additional positive or negative impulses are required to explain the behavior). If the curves intersect at the times $\tau,\nu, \mu$, the relation between $b_1$, $b_2$ and $b_1^*$, $b_2^*$ is given by a solution to the equation
\begin{multline}\label{eq:3sampl}
    (\chi^\tau-\psi^\tau)(\omega^{\mu}-\omega^{\nu}) + (\chi^\nu-\psi^\nu)(\omega^{\tau}-\omega^{\mu})\\
    + (\chi^\mu-\psi^\mu)(\omega^{\nu}-\omega^{\tau}) = 0,
\end{multline}
where $\chi=\e^{b_1 - b_2^*}$, $\psi=\e^{b_1 - b_1^*}$ and $\omega=\e^{b_1 - b_2}$.
%\begin{equation*}
Note that equation \eqref{eq:3sampl} has $b_1=b_2$ as another, infeasible solution. By solving \eqref{eq:3sampl} for $\omega$, an expression for $b_2$ would be obtained. However, since solving \eqref{eq:3sampl} algebraically for $\omega$ is not possible in a general case, only equidistant sampling, i.e. $\nu=\tau+c, \mu=\tau+2c$, where $c>0$, is considered. The feasible solution then becomes
\begin{multline*}
    b_2 = b_1 - \ln (\omega)\\
    = b_1 - \frac{1}{c}\ln \Big(\frac{\chi^{\tau+c}-\chi^{\tau+2c}-\psi^{\tau+c}+\psi^{\tau+2c}}{\chi^\tau-\chi^{\tau+c}-\psi^\tau+\psi^{\tau+c}}\Big)\\
    \triangleq b_1 - \frac{1}{c}\ln \Big(\frac{\omega_1}{\omega_2}\Big),
\end{multline*}
where the solution naturally is $b_2=b_2^*$ if $b_1=b_1^*$. Taking the derivative with respect to $b_1$ yields
\begin{equation*}
    \frac{d b_2}{d b_1} = -1 + (\chi^{\tau+c} - \psi^{\tau+c})\Big(\frac{1}{\omega_1}-\frac{1}{\omega_2}\Big).
\end{equation*}
Since $b_1<b_2$, it follows that $\chi < \psi$ and $0<\omega_1<\omega_2$. That leads to the inequality
\begin{equation*}
    -\infty < \frac{d b_2}{d b_1} < -1.
\end{equation*}
The second derivative, given by
\begin{multline*}
    \frac{d^2 b_2}{d b_1^2}=c \big(\chi^{\tau+c} - \psi^{\tau+c}\big)\\
    \times \Big(-\frac{1}{\omega_1}-\frac{1}{\omega_2} + \big(\chi^{\tau+c} - \psi^{\tau+c}\big) \Big(\frac{1}{\omega_1^2}-\frac{1}{\omega_2^2} \Big) \Big)>0,
\end{multline*}
shows that the derivative changes monotonously. Finally, if the samples are shifted in time relative to the impulse, $b_2$ changes according to
\begin{equation*}
    \frac{d b_2}{d\tau} = \frac{1}{c}(b_1-b_2)\psi^{\tau+c}(1-\psi^c) \big(\frac{1}{\omega_2}-\frac{1}{\omega_1}\big).
\end{equation*}
If $b_1<b_1^*$, then $\psi<1$, which makes the expression positive, while $b_1>b_1^*$ makes it negative, i.e. a shift in time causes a pivot of the curve around the point $(b_1^*,b_2^*)$. The resulting curves are illustrated in Fig. \ref{fig:bound}.

\begin{figure}
\begin{center}
\vspace{0.2cm}
\includegraphics[trim={0.7cm 0 0.7cm 0.4cm}, clip, scale=0.64]{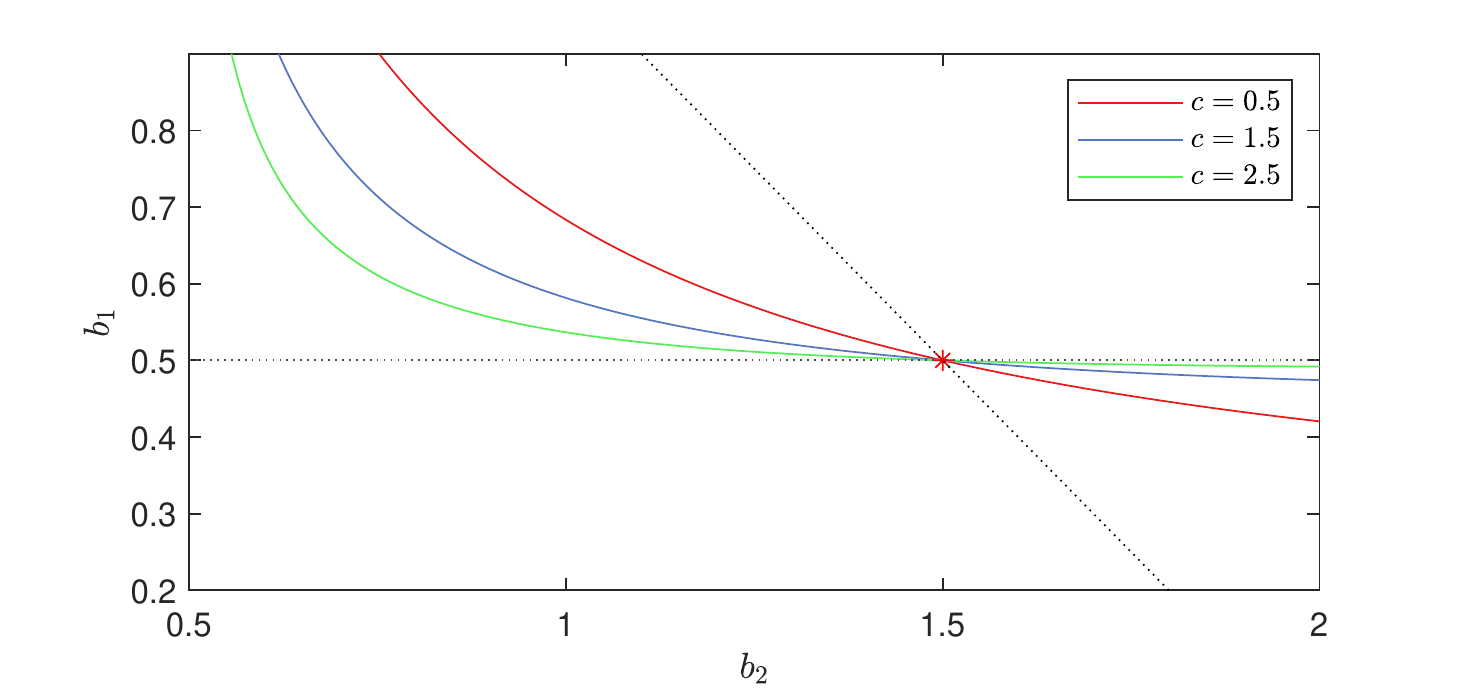}
\caption{The pairs of estimates $b_1, b_2$ resulting in intersections between $y(t)$ and $y^*(t)$ at times $1,1+c,1+2c$ for different values of $c$ and $b_1^*=0.5$, $b_2^*=1.5$. The dotted lines represent the theoretical limits for the borders given in Proposition \ref{prop:impsign}.}
\label{fig:bound}
\end{center}
\end{figure}

If more than three measurements are considered, each triplet of measurements  generates a separate $b_2$-boundary. For a non-equidistant triplet, the qualitative behavior of the corresponding solution to \eqref{eq:3sampl} appears to be similar to that of the equidistant case in numerical experiments. Since a solution with only positive impulses requires that parameter estimates are above all these boundaries, the limiting case, which defines the curve $\gamma_P$, corresponds to the largest value of $b_2$ for a given $b_1$. For solutions with almost all impulses being negative, the converse holds. Adding more impulses does not alter this behavior. As all boundaries intersect at $b_1=b_1^*$, $b_2=b_2^*$, and triplets with different time shifts produce different slopes according to the above discussion, this explains the non-smoothness of the boundaries at this point illustrated in Fig.~\ref{fig:optSol}. Note that the presented analysis only is valid if the impulses are sufficiently sparse, i.e. there are instances with at least three sampling times in between two impulses.

\subsection{Estimation algorithm}\label{sec:nonoiseest}
The following procedure is suggested to solve the estimation problem in the absence of noise. Let $d(c)$ denote the distance between the pair of points where the boundaries of the positive and negative impulse regions (see Fig.~\ref{fig:optSol}) intersect with the line $b_1=2b_2+c$. Since the boundaries meet at the point $(b_1^*,b_2^*)$, minimization of $d(c)$ with respect to $c$ recovers the values of $b_1^*$ and $b_2^*$.
%the values of $b_1^*$ and $b_2^*$ can be obtained by minimizing $d(c)$ using standard optimization, provided that the starting guess is close enough to the optimum.
The points of intersection which determine $d(c)$ can be calculated with the bisection method, by considering the signs of the estimated impulse weights obtained from solving \eqref{eq:opt} at points along the line $b_1=2b_2+c$. Once $b_1^*$ and $b_2^*$ are calculated, Algorithm $1$ in \cite{MM:13} is used to calculate the impulse times and weights.

No further details on this algorithm are provided, as only estimation under uncertainties is relevant in applications.

\section{ESTIMATION UNDER UNCERTAINTY}\label{sec:noisy}

\subsection{Estimation of $b_1$ and $b_2$}
To represent noise or model uncertainty, we consider the modification of \eqref{eq:matrixY}
\begin{equation*}
    Y = \Phi(b_1, b_2) \theta + \epsilon,
\end{equation*}
where $\epsilon$ is a zero-mean noise vector. The parameter space is then no longer divided as in Fig.~\ref{fig:optSol}, since the region of mixed impulse weights then covers a larger area and, in particular, includes $(b_1^*,b_2^*)$. A non-negativity constraint for the impulse weights is for this reason added to \eqref{eq:opt}. In the absence of noise, this would leave the region above $\gamma_P$ unaffected, while rendering the residual sum of squares
\begin{equation}
    \|Y - \Phi(b_1, b_2) \hat\theta(b_1, b_2) \|^2 \triangleq g(b_1,b_2)\label{eq:f2}
\end{equation}
nonzero in the rest of the parameter space. In particular, for a fixed value $b_2=\bar b_2$ such that $(\bar b_1,\bar b_2)$ is on $\gamma_P$, $g(b_1,\bar b_2)$ is expected to be decreasing in $b_1$ when $b_1<\bar b_1$. In the noisy setting, the qualitative behavior tends to be similar, but the residual sum is nonzero even for $b_1>\bar b_1$. We will utilize this property to estimate $\gamma_P$ using the following result.
\begin{lemma}\label{lem:quad}
Let $f(x) = c_1(x-x^*)^2+c_2$, where $c_1,c_2>0$. Define $N_{f}(x)$ and $\hat x$ by
\begin{equation*}
    N_{f}(x) = -\frac{f(x)}{df(x)/dx},
\end{equation*}
\begin{align*} 
    \hat x =& \argmin_{x} N_{f}(x) + \min_{x} N_{f}(x),\\
    \textrm{s.t. } & \min_{x} N_{f}(x)>0.
\end{align*}
Then $\hat x=x^*$.
\end{lemma}
\begin{proof}
Straightforward minimization.
\end{proof}
The minimization problem stated in the lemma is applied to the residual sum of squares $g(b_1,\bar b_2)$, i.e. the function is assumed to be approximately quadratic, somewhat similarly to Newton's method in optimization. The point of this technique is that the objective function only is required to be quadratic below $\gamma_P$, so the optimization can be performed even though $g(b_1,\bar b_2)$ does not have a unique minimum close to $\gamma_P$.

Generalizing to incorporate both $b_1$ and $b_2$ in the formulation, and constraining the permitted number of impulses (mimicking the effect of the constraint on $\min N_f$ above), we arrive at the formulation
\begin{equation}\label{eq:b1b2opt}
\begin{aligned}
    (\hat b_1, \hat b_2) =& \argmin_{b_1,b_2} N_{g}(b_1,b_2) + \min_{b_1,b_2} N_{g}(b_1,b_2),\\
    \textrm{s.t. } & \#\{d_k>d_\mathrm{min}\}\le \Pi,
\end{aligned}
\end{equation}
where 
\begin{equation*}
N_g=\frac{-g(b_1,b_2)}{\partial g(b_1,b_2) /\partial b_1},
\end{equation*}
$\Pi$ is the maximal permitted number of impulses, $d_\mathrm{min}$ is the threshold for counting an impulse and $\#$ denotes cardinality. It should be noted that the cost function, together with the impulse number constraint, constitute a nonconvex optimization problem which admits multiple local minima.

It is however not obvious that a solution to \eqref{eq:b1b2opt} approximates $(b_1^*,b_2^*)$, and not some other point on $\gamma_P$. Indeed, the numerical experiments in the next section demonstrate that the estimation works well with relatively low levels of noise and a frequent sampling, but with a higher noise level, only an estimate $\hat \gamma_P$ of this boundary can be found.
The corresponding optimization formulation
%to estimate $\gamma_P$
then becomes
\begin{equation}\label{eq:b1opt}
\begin{aligned}
    \bar b_1 =& \argmin_{b_1} N_{g}(b_1,\bar b_2) + \min_{b_1} N_{g}(b_1,\bar b_2),\\
    \textrm{s.t. } & \#\{d_k>d_\mathrm{min}\}\le \Pi,
\end{aligned}
\end{equation}
where $(\bar b_1,\bar b_2)$ is the intersection between $\hat \gamma_P$ and $b_2=\bar b_2$.

\subsection{Estimation algorithms}
As the optimization problems \eqref{eq:b1b2opt}, \eqref{eq:b1opt} are non-convex, we use gridding to solve them, utilizing the finite difference over the grid points to approximate the derivative in $N_g$. That means that, in the low noise case, the parameter combination in the grid which minimizes $N_g$ is used in the calculation, while in the high noise case, for each $b_2$-grid point, the minimizing $b_1$ is used. Finally, to determine the location and weights of the impulses, all impulses below a user-defined threshold are removed and adjacent impulses are merged. The resulting procedure is summarized in Algorithm~\ref{alg:est}.
\begin{algorithm}
\begin{algorithmic}[1]
\State Calculate $(\hat b_1,\hat b_2)$ from \eqref{eq:b1b2opt} (low noise) \textit{or} solve \eqref{eq:b1opt} for $\bar b_1=\hat b_1$ for a given $\bar b_2=\hat b_2$ (high noise)
\State Calculate $\hat \theta(\hat b_1,\hat b_2)$ from \eqref{eq:opt}
\State Let $S=\{k\in\{1,\dots,K\} \mid \hat d_k<d_\mathrm{min}\}$
\State Solve \eqref{eq:opt} with all $d_k$ with $k\in S$ constrained to be zero
\State Merge adjacent non-zero impulses according to Algorithm $1$ in \cite{MM:13}
\caption{Impulse and time constant estimation}\label{alg:est}
\end{algorithmic}
\end{algorithm}

\section{NUMERICAL EXPERIMENTS}\label{sec:numexp}
The proposed estimation technique is evaluated on synthetic data consisting of the response to three impulses with additive zero-mean Gaussian measurement noise. Two Monte Carlo experiments with $100$ realizations were performed; one with low noise variance (experiment $A$) and one with a more realistic (i.e. higher) noise level (experiment $B$). In the former case, a comparison is made with the implementation in \cite{MM:13}. The parameters of the experiments are specified in Table~\ref{tab:MCpars}, while examples of data realizations are shown in Fig. \ref{fig:outputEx}. The parameter values of the estimation algorithms are given in Table~\ref{tab:estpars}.

\begin{figure}
\begin{center}
\vspace{0.2cm}
\includegraphics[trim={0.7cm 0 0.7cm 0.4cm}, clip, scale=0.64]{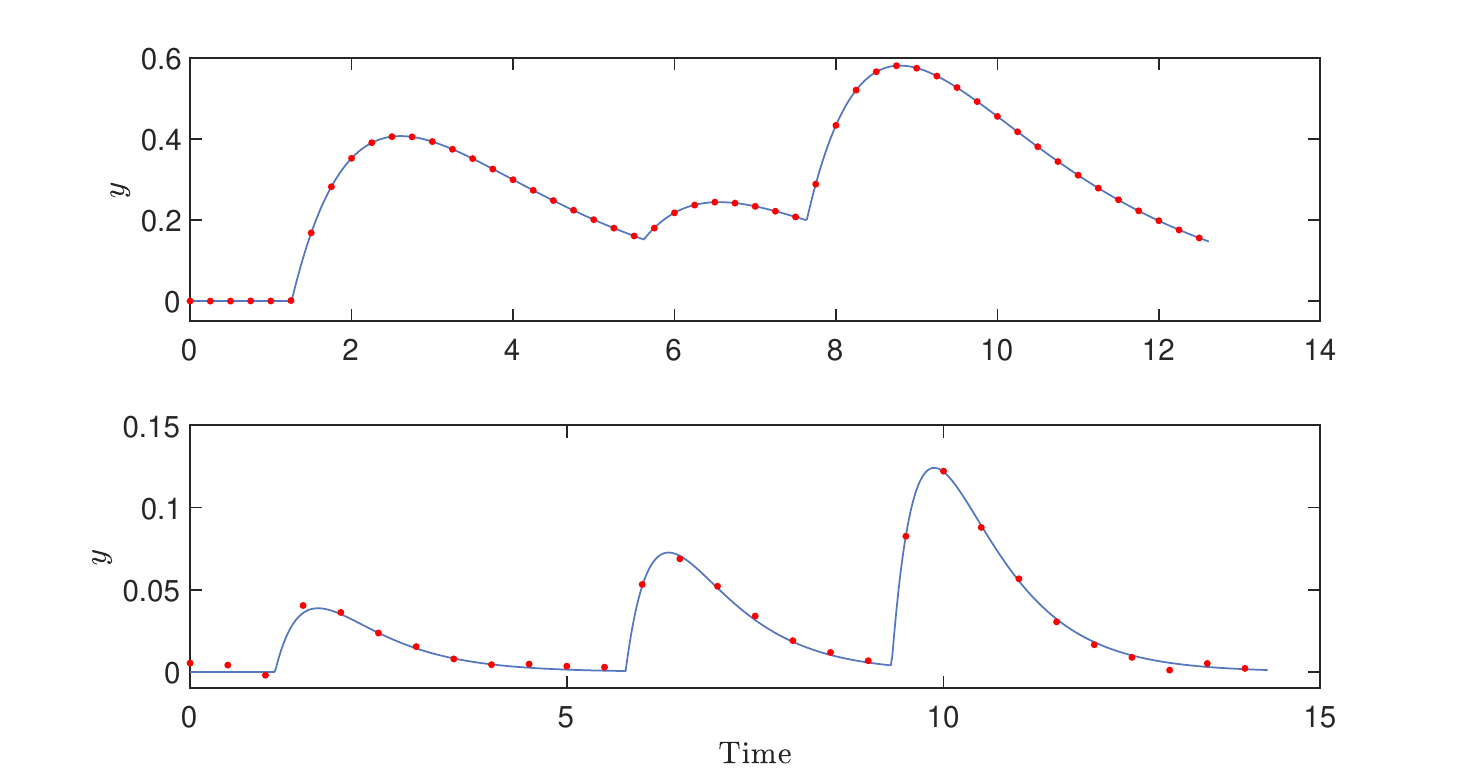}
\caption{Examples of sampled synthetic data with low (top) and high (bottom) levels of additive measurement noise.}
\label{fig:outputEx}
\end{center}
\end{figure}

\begin{table}[ht]
\caption{Monte Carlo experiment parameter distributions (left) and fixed values (right). $U_{[\cdot,\cdot]}$ denotes uniform distribution, $\Delta \tau_i$ and $\Delta t_i$ are the time separation between consecutive impulses and samples respectively, $\sigma$ is the noise standard deviation and $\tau_\mathrm{end}$ is the time between the last impulse and the end of the time horizon.}
\label{tab:MCpars}
\begin{center}
\begin{tabular}{|c|c|}
\hline
$b_1^*$ & $U_{[0.4,1.4]}$ \\
$b_2^*-b_1^*$ & $U_{[0.3,1.3]}$ \\
$d_i^*$ & $U_{[0.1,1]}$ \\
$\Delta \tau_i$ & $U_{[1,5]}$ \\
\hline
\end{tabular}
\quad
\begin{tabular}{|c|cc|}
\hline
 & $A$ & $B$ \\
\hline
$\sigma$ & $\num{2e-4}$ & $0.0015$ \\
$\Delta t_i$ & $0.25$ & $0.5$ \\
%$\#\{d_i^*\}$ & 3 & 3 \\%\multicolumn{2}{c|}{$3$} \\
$\tau_\mathrm{end}$ & 5 & 5 \\%\multicolumn{2}{c|}{$5$} \\
\hline
\end{tabular}
\end{center}
\end{table}

\begin{table}[ht]
\vspace{0.2cm}
\caption{Estimation algorithm parameters. $\Delta b$ is the distance between the grid points and $\bar d$ denotes the mean impulse weight.}
\label{tab:estpars}
\begin{center}
\begin{tabular}{|c|c|}
\hline
$b_1$ range & $[0.5 b_1^*,0.5(b_1^*+b_2^*))$ \\
$b_2$ range & $[0.5(b_1^*+b_2^*),1.5b_2^*]$\\
$\Delta b$ & 0.02 \\
$d_\mathrm{min}$ & $0.05\bar d$ \\
$\Pi$ & $0.5K$ \\
\hline
\end{tabular}
\end{center}
\end{table}

\subsection{Estimation from low noise-data}
The ``low noise'' version of Algorithm \ref{alg:est} was used on synthetic data with additive measurement noise of low variance. The results are compared with the $\ell_1$-constrained estimation algorithm in \cite{MM:13}, implemented with the same $b_1$-$b_2$-grid. To account for the regularization that is required in that method, gridding was also performed over possible values of the regularization parameter $\lambda_\mathrm{max}$, and the Akaike information criterion was used to determine its value. As displayed in Table~\ref{tab:1stand2nd}, the current implementation performs better.

\begin{table}[ht]
\caption{Root mean squared errors for parameters estimated using Algorithm~\ref{alg:est} ($A1$) and regularized LS \cite{MM:13} ($A2$).}
\label{tab:1stand2nd}
\begin{center}
\begin{tabular}{|c|cc|}
\hline
 & $A1$ &  $A2$ \\
\hline
$b_1$ & $0.0105$ & $0.0234$ \\
$b_2$ & $0.0255$ & $0.0582$ \\
$d_i$ & $0.0164$ &  \\
$\tau_i$ & $0.0745$ &  \\
\hline
\end{tabular}
\end{center}
\end{table}

The estimated input was also evaluated. In $78$\% of the realizations, a correct number of impulses were estimated whereas the remaining had an average of $1.77$ extra impulses caused by the noise. The impulses corresponding to the true input sequences were identified, and the resulting estimation errors are given in Table \ref{tab:1stand2nd}. Both the timing and weights of the impulses are estimated with satisfactory accuracy.

\subsection{Estimation from realistic data}
Here the ``high noise'' version of Algorithm \ref{alg:est} was employed on data with a higher noise level, and less frequent sampling than the previous case, over a grid of $b_2$-values. In Fig.~\ref{fig:2ndorderboundary}, the estimate $\hat \gamma_P$ for one data set is displayed, together with the true parameter values and $\gamma_P$. The resulting average Euclidean distance between $\hat \gamma_P$ and $(b_1^*,b_2^*)$ over all realizations is $0.0122$. The curve evidently tends to be close to the true parameters, but a strategy to obtain the best estimate along this curve has not been found.

\begin{figure}[ht]
\begin{center}
\vspace{0.2cm}
\includegraphics[trim={0.7cm 0 0.7cm 0.4cm}, clip, scale=0.64]{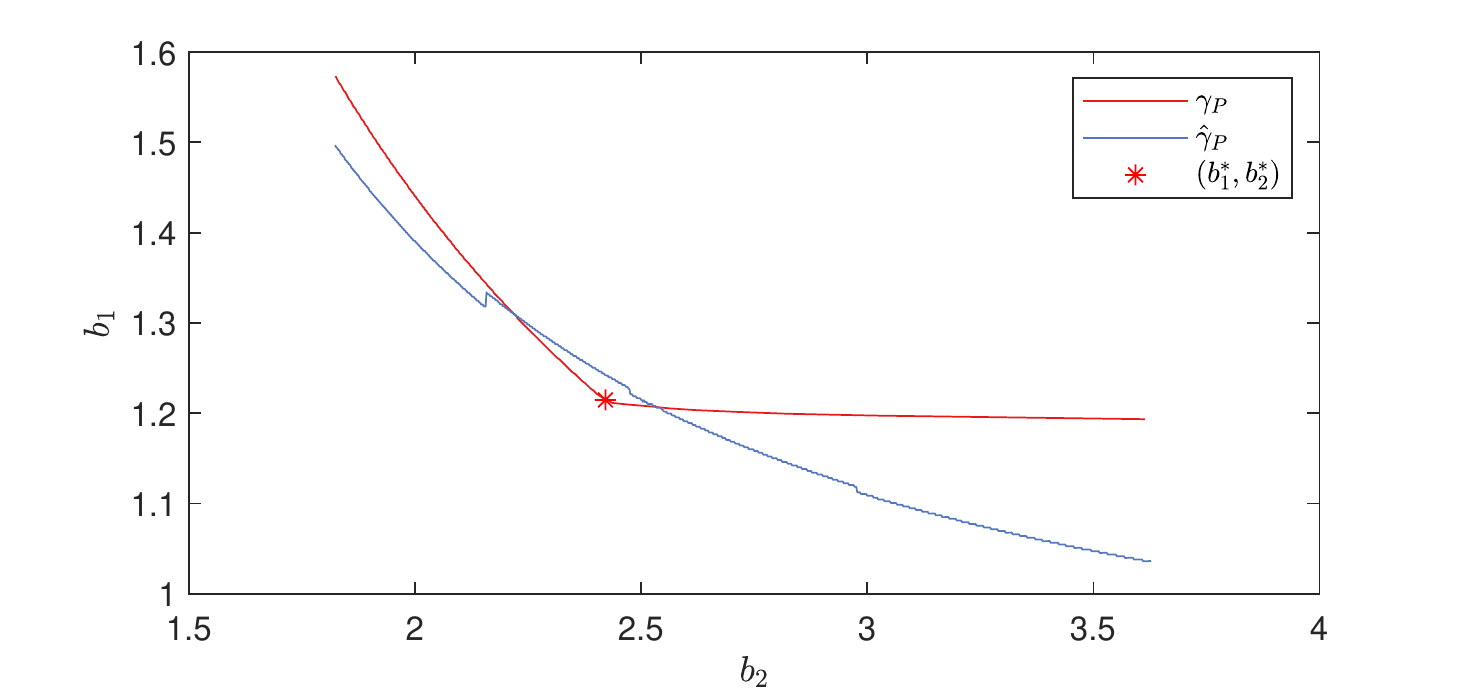}
\caption{The $\gamma_P$-boundary, its estimate $\hat \gamma_P$ generated using Algorithm \ref{alg:est} (with a higher grid resolution for improved visual appearance), and the true parameter values for one synthetic data realization.}
\label{fig:2ndorderboundary}
\end{center}
\end{figure}

\section{CONCLUSIONS}
A novel estimation technique for a class of continuous second-order systems with impulsive input has been presented. It builds upon previous work \cite{MM:13}, but outperforms that method on synthetic data with low noise and also has the advantage of not requiring any data-dependent user-defined parameters. Under strong uncertainty, only an implicit relation between the plant parameters can be determined. It is hypothesized that under such circumstances, unique parameter values generally cannot be reliably estimated. A more thorough analysis of this issue, and its implications on estimations from clinical data, are possible future research directions.

%%%%%%%%%%%%%%%%%%%%%%%%%%%%%%%%%%%%%%%%%%%%%%%%%%%%%%%%%%%%%%%%%%%%%%%%%%%%%%%%

%%%%%%%%%%%%%%%%%%%%%%%%%%%%%%%%%%%%%%%%%%%%%%%%%%%%%%%%%%%%%%%%%%%%%%%%%%%%%%%%

%%%%%%%%%%%%%%%%%%%%%%%%%%%%%%%%%%%%%%%%%%%%%%%%%%%%%%%%%%%%%%%%%%%%%%%%%%%%%%%%
\appendices

\section{Proof of Lemma \ref{lem:intersect}} \label{app:intersect}
Let $\hat x_i, i=1,2$ denote the states of the estimated system. First, consider the case of the two outputs being tangent
%\textcolor{green}{Should it be parallel} \textcolor{red}{No}
for some $\tau>\max \{t_1, \hat t_1\}$, i.e. $x_2(\tau)=\hat x_2(\tau)$ and $\dot x_2(\tau) = \dot{\hat x}_2(\tau)$. The dynamics for $x_2$ and $\hat x_2$ then gives
\begin{equation} \label{eq:x1hat}
    \hat x_1(\tau) = x_1(\tau)+(\hat b_2 - b_2) x_2(\tau),
\end{equation}
which is used to calculate second derivative:
\begin{equation}\label{eq:x2dotdot}
\begin{aligned}
        \ddot{\hat x}_2(\tau)-\ddot x_2(\tau) = \dot{\hat x}_1(\tau) - \hat b_2 \dot{\hat x}_2(\tau) - \dot x_1(\tau) + b_2 \dot x_2(\tau)\\
    %= -\hat b_1 \hat x_1(\tau) - \hat b_2(x_1(\tau)-b_2 x_2(\tau))\\
    %+ b_1 x_1(\tau) + b_2(x_1(\tau) - b_2 x_2(\tau))\\
    = (b_1 - \hat b_1 + b_2 - \hat b_2) x_1(\tau) + (b_2-\hat b_1) (\hat b_2 - b_2) x_2(\tau).
\end{aligned}
\end{equation}
Two separate cases establish the sign of this expression:
\begin{itemize}
    \item If $\hat b_2 - b_2 \le 0$ then $(b_2-\hat b_1) (\hat b_2 - b_2) \le 0$,
    %(since $\hat b_2 - b_2 \le 0$ implies $\hat b_1 - b_2 \le 0$),
    so $\ddot{\hat x}_2(\tau)-\ddot x_2(\tau)$ is negative since $(b_1 - \hat b_1 + b_2 - \hat b_2)<0$;
    %\textcolor{green}{R9: this depend on the sign of $(b_2-\hat b_1)$ which is unknown}
    \item If $\hat b_2 - b_2 > 0$, then \eqref{eq:x1hat} implies $\hat x_1(\tau)>x_1(\tau)$ and thus $\dot{\hat x}_1(\tau) = -\hat b_1 \hat x_1(\tau) < -b_1 x_1(\tau) = \dot x_1(\tau)$. Applying this to the first row of \eqref{eq:x2dotdot} gives $\ddot{\hat x}_2(\tau)-\ddot x_2(\tau)<0$.
\end{itemize}
It follows that $\hat x_2(t) \le x_2(t)$ in a neighborhood of $\tau$.

%\addtolength{\textheight}{-1.5cm}   % This command serves to balance the column lengths

Now assume that $x_2(\nu)=\hat x_2(\nu)$, for some $\nu \ne \tau$ and that there are no intersections between $x_2(t)$ and $\hat x_2(t)$ for $t$ between $\tau$ and $\nu$ (otherwise use that intersection to define $\nu$). Using \eqref{eq:x1hat} and the dynamics of the first states
\begin{equation*}
    x_1(t) = d_1 \e^{-b_1 (t-t_1)}, \quad \hat x_1(t) = \hat d_1 \e^{-\hat b_1 (t-\hat t_1)},
\end{equation*}
an expression for $\dot{\hat x}_2(\nu) - \dot x_2(\nu)$ can be calculated as
\begin{multline*}
    \dot{\hat x}_2(\nu) - \dot x_2(\nu) = (b_2-\hat b_2) x_2(\nu)+\hat x_1(\nu)-x_1(\nu)\\
    = (\hat b_2 - b_2) \big(x_2(\tau) \e^{-\hat b_1(\nu-\tau)} - x_2(\nu)\big)\\
    + x_1(\tau)\big(\e^{-\hat b_1(\nu-\tau)} - \e^{-b_1(\nu-\tau)}\big)\\
    = \frac{d_1}{b_1-b_2}\Big((\hat b_2-b_2)\e^{-b_2\nu}(\e^{(\nu-\tau)(b_2-\hat b_1)}-1)\\
    +(b_1-\hat b_2)\e^{-b_1\nu}(\e^{(\nu-\tau)(b_1-\hat b_1)}-1)\Big).
\end{multline*}
Since $\dot{\hat x}_2(\nu) - \dot x_2(\nu)=0$ for $\nu=\tau$, the sign is established by the derivative
\begin{multline*}
    \frac{d}{d\tau}(\dot{\hat x}_2(\nu) - \dot x_2(\nu))\\
    = \frac{d_1}{b_1-b_2}\Big((\hat b_2-b_2)(\hat b_1-b_2)\e^{-\hat b_1(\nu-\tau)-b_2 \tau}\\
    + (b_1-\hat b_2)(\hat b_1-b_1)\e^{-\hat b_1(\nu-\tau)-b_1 \tau}\Big).
\end{multline*}
Now use $\hat b_2-b_2 > b_1-\hat b_1$, $\hat b_1 -b_2 > b_1-\hat b_2$ to get% the inequality 
\begin{multline*}
    \frac{d}{d\tau}(\dot{\hat x}_2(\nu) - \dot x_2(\nu))\\
    >  \frac{d_1}{b_1-b_2}\Big((b_1-\hat b_1)(b_1-\hat b_2)\e^{-\hat b_1(\nu-\tau)-b_2 \tau}\\
    + (b_1-\hat b_2)(\hat b_1-b_1)\e^{-\hat b_1(\nu-\tau)-b_1 \tau}\Big)\\
     = \frac{d_1}{b_1-b_2}(b_1-\hat b_1)(b_1-\hat b_2)\e^{-\hat b_1(\nu-\tau)}\big(\e^{-\tau b_2} - \e^{-\tau b_1} \big)>0.
\end{multline*}

This implies that $\dot{\hat x}_2(\nu) - \dot x_2(\nu)$ is positive when $\tau>\nu$, which in turn implies that $\hat x_2(t)>x_2(t)$ for $t>\nu$ and close to $\nu$. Since $x_2(t)$ and $\hat x_2(t)$ are continuous and there are no intersections between $x_2(t)$ and $\hat x_2(t)$ for $\nu<t<\tau$, this is contradictory with $\hat x_2(t) \le x_2(t)$ for $t$ in a neighborhood of $\tau$. Since the case $\tau<\nu$ leads to a contradiction in the same way, one can conclude that if $x_2(\tau)=\hat x_2(\tau)$ and $\dot x_2(\tau) = \dot{\hat x}_2(\tau)$, a $\nu \ne \tau$ such that $x_2(\nu)=\hat x_2(\nu)$ does not exist.

Now suppose that there are more than two intersections between $x_2(t)$ and $\hat x_2(t)$. Since $\hat x_2(t)$ depends linearly on $\hat d_1$, it is then possible to reduce this weight until two of the intersections are reduced to one tangent point, while other intersections still exist (or possibly also are reduced to tangent points). But in the tangent case, there can be no other intersections between the curves, so there cannot be more than two intersections between $x_2(t)$ and $\hat x_2(t)$.

\section{Proof of Proposition \ref{prop:impsign}}\label{app:impsign}
We only show the case $b_1+b_2>b_1^*+b_2^*$ and $b_1>b_1^*$, as a similar technique can be used in the other case.

Let $x_i^*,i=1,2$ and $x_i,i=1,2$ respectively denote the states of the true and  the estimated system. Consider the output of the estimated system, if it were driven by the same impulses as the true system. Since $x_2(t)$ depends monotonously on both $b_1-b_2$ and $b_2$, it then follows that $x_2(t)<x_2^*(t)$ for all $t>t_1$. Since $x_2(t)$ depends linearly on the impulse weights, the weight of the first impulse $d_1$ can be increased so that $x_2(t_2)=x_2^*(t_2)$. Utilizing the asymptotic behavior of the systems and Lemma \ref{lem:intersect}, it can be shown that in this case $x_2(t)<x_2^*(t)$ for $t>t_2$. Now apply the same technique for the whole optimization horizon, i.e. at every $t_k$, increase the impulse weight (generally from zero) so that $x_2(t_{k+1})=x_2^*(t_{k+1})$. Since this results in a solution that is optimal, and the solution is unique (see \cite{MM:13}), it follows that all estimated impulses are positive.

%%%%%%%%%%%%%%%%%%%%%%%%%%%%%%%%%%%%%%%%%%%%%%%%%%%%%%%%%%%%%%%%%%%%%%%%%%%%%%%%

\bibliographystyle{IEEEtran}
\bibliography{IEEEabrv,refs}

\end{document}